\documentclass[12pt,a4paper,oneside]{article}
\usepackage[centertags]{amsmath}%

\usepackage[utf8]{inputenc}
\usepackage[english]{babel}
\usepackage{mathtools}
\usepackage{amsfonts}
\usepackage{amssymb}
\usepackage{amsthm}
\usepackage{dsfont}
\usepackage{color}
\newtheorem{definition}{Definition}
\newtheorem{theorem}[definition]{Theorem}
\newtheorem{proposition}[definition]{Proposition}
\newtheorem{lemma}[definition]{Lemma}
\newtheorem{corollary}[definition]{Corollary}
\newtheorem{rem}[definition]{Remark}
\newenvironment{remark}{\begin{rem}  \rm }{\end{rem}}
\newtheorem{rems}[definition]{Remarks}

\newtheorem{example}[definition]{Example}

\newcommand{\mi}{{\mathrm{i}}}
\newcommand\unit{\hbox{\rm 1\kern-2.8truept l}} %identity of the algebra
\newcommand\re{\Re \kern-1.4truept e}
\newcommand\im{\Im \kern-1.4truept m}

\newcommand\Ran{{\rm{Ran\,}}}

\newcommand\spanno{{\rm{span}}}

\newcommand\M{\mathcal{M}}
\def\rel3{\alpha}

\newcommand\de{{\rm{d}}}
\newcommand{\h}{\mathsf{h}} %hilbert space
\newcommand{\kk}{\mathsf{k}} %hilbert space
\newcommand{\mm}{\mathsf{m}}
 
\newcommand{\T}{\mathcal{T}} %semigroup 

\newcommand{\dissT}{\T^{{\rm da}}}
\newcommand{\dissL}{\Ll^{{\rm da}}}
\newcommand{\dfT}{\T^{{\rm df}}}
\newcommand{\dfL}{\Ll^{{\rm df}}}
\newcommand{\Ll}{\mathcal{L}} %generator of the QMS

\newcommand{\Pmin}[1]{{\mathcal{P}_{\rm min}}(#1)}  %minimal prj
\newcommand{\Z}[1]{{\mathcal{Z}}(#1)}  % center 
\newcommand{\nT}{\mathcal{N(\mathcal{T})}} %N(T)
\newcommand{\FT}{\mathcal{F(\mathcal{T})}} %fixed points
\newcommand{\B}{\mathcal{B}(\h) } %B(H)=linear and bounded operators on H

\newcommand{\E}{\mathcal{E}} %conditional expectation
\newcommand{\R}{\mathbb{R}} %real numbers
\newcommand{\C}{\mathbb{C}} %complex numbers
\newcommand{\nm}{\|} %norm
\newcommand{\Ss}{\mathsf{s}} % S
\newcommand{\ff}{\mathsf{f}} % F

\newcommand{\Rr}[1]{\mathcal{R}(#1)} %reversible space

\newcommand{\scal}[2]{\langle{#1},{#2}\rangle} %scalar product
\newcommand{\ee}[2]{|{e_{#1}}\rangle\langle{e_{#2}}|} %bracket
\newcommand{\kb}[2]{|{#1}\rangle\langle{#2}|}

\def\red{\color{red}}
\def\blu{\color{blue}}

\newcommand{\tr}[1]{{\rm tr }(#1)}   %#1: argument of the trace

\begin{document}

\title{The general structure of the Decoherence-free subalgebra for uniformly continuous Quantum Markov semigroups
}
%\subtitle{Do you have a subtitle?\\ If so, write it here}

%\titlerunning{Short form of title}        % if too long for running head

\author{Emanuela Sasso        \and
        Veronica Umanit\`a %et
}

\maketitle

\begin{abstract}
By using the decomposition of the decoherence-free subalgebra $\nT$ in direct integrals of factors, we obtain a structure theorem for every uniformly continuous QMSs. Moreover we prove that, when there exists a faithful normal invariant state, $\nT$ has to be atomic and decoherence takes place.
%\keywords{Quantum Markov semigroup \and Decoherence-free subalgebra \and atomicity \and reversible algebra \and direct integrals of von Neumann algebras}
% \PACS{PACS code1 \and PACS code2 \and more}
% \subclass{MSC code1 \and MSC code2 \and more}
\end{abstract}

\section{Introduction}
The irreversible evolution of an Open Quantum System can be described by a Quantum Markov Semigroups (QMS) on a von Neumann algebra, i.e. a weakly$^*$ continuous semigroup $\T=(\T_t)_{t\geq 0}$ of normal, completely positive, identity preserving maps. The loss of $*-$automorphic dynamics, due to the interaction of the system with the environment, is often indicated by physicists with the name of \lq\lq decoherence\rq\rq. In order to reduce this phenomenon it is important to identify a \emph{decoherence-free subalgebra}, i.e. a maximal subalgebra on which the system behaves as a closed system (see \cite{LW,KL,BlOl,O,O2,TV}). In other words, the decoherence-free subalgebra is noisy-free with respect to the typical effects of the interaction. From a mathematical point of view this idea is expressed in a strengthened form in the definition of environmental decoherence introduced by Blanchard and Olkievicz in \cite{BlOl} (see also \cite{Hellmich,CSU2,CSU-new}), because it requires that non-automorphic part of the evolution vanishes in time. \\
For a special class of open quantum systems, described by uniformly continuous QMSs on $\B$ (the space of all bounded operators on a complex separable Hilbert space $\h$), the decoherence-free subalgebra is uniquely determined and it coincides with $\nT$, the biggest von Neumann algebra on which the semigroups $\T$ acts as a $*-$automorphism (see  \cite{DFR,FFRR,FV82}). In \cite{DFSU,FSU-NtFt} we showed that, whenever $\nT$ is atomic,  i.e. it can be written as direct sum of type $I$ factors, its block diagonal structure forces the Lindblad operators to have a diagonal form, and induces a decomposition of the system into its noiseless and purely dissipative part, also determining the structure of invariant states. It is then natural to ask what happens if $\nT$ is not atomic, if it is however possible to provide an explicit caracterization of this algebra.

The atomicity property of the decoherence-free subalgebra plays a crucial role also in the study of environmental decoherence, i.e. in the possibility to decompose the algebra $\B$ into the direct sum of $\nT$ and a remaining space on which the dynamics vanishing in time: indeed we proved in \cite[Theorem 11]{FSU-NtFt} that, assuming the existence of a faithful normal invariant state,  the atomicity of $\nT$ is equivalent to have the previous splitting with $\nT$ equal to the so-called \emph{reversible algebra} $\mathfrak{M}_r$, i.e. the weak$^*$ closure of the algebra generated by the eigenvectors of $\T_t$ corresponding to zero real part eigenvalues. This algebra appears in another asymptotic decomposition of $\B$, the Jacobs-de Leeuw-Glicksberg splitting (\cite{batkai,Hellmich}), and it is always the image of a normal conditional expectation. This last property ensures that $\mathfrak{M}_r$ is atomic.
So it is quite natural to compare the decoherence-free subalgebra and the reversible algebra also in a more general context, i.e. when $\nT$ is not necessarily atomic. 

In this paper we deepen the analysis of $\nT$ providing the following results:
\begin{enumerate}
\item $\nT$ is decomposable with respect a suitable direct integral representation of $\h$, and its structure influences that one of the infinitesimal generator of $\T$. The starting idea has been to substitute the direct sum of type $I$ factors in the decomposition of an atomic algebra, with the splitting in direct integral of factors, through the central decomposition induced by the center of $\nT$. 

\item If the QMS has a faithful normal invariant state, $\nT$ coincides with the reversible algebra $\mathfrak{M}_r$ and so it has to be atomic.
\end{enumerate}
This last result has many interesting consequences on the study of QMSs, since it allows to solve some open problems concerning environmental decoherence. Indeed we obtain that, under the assumption of the existence of a faithful invariant state, the decomposition of $\B$ induced by decoherence always exists, it is uniquely determined and it coincides with the Jacobs-de Leeuw-Glicksberg splitting. This is exactly the same result we found in the finite-dimensional case (see \cite{CSU2}).

The paper is organized in the following way. In the second section we recall some basic facts on the structure of a uniformly continuous QMS with atomic decoherence-free subalgebra. In order to show that not all decoherence-free subalgebras are atomic, we exhibit an example in which $\nT$ is a type $II_1$ factor. In Section $3$ we introduce the decomposition of $\nT$ in direct integral of factors and determine the structure induced on the infinitesimal generator. Section $4$ contains the main results of the paper: the atomicity of $\nT$ when there exists a faithful normal invariant state and its equivalence with the reversible algebra. {We remind in Appendix some results about the theory of direct integrals of von Neumann algebras.}

\section{The decoherence-free subalgebra $\nT$}

Let $\h$ be a complex separable Hilbert space. A QMS on 
the algebra $\B$ of all linear and bounded operators on $\h$ is a weakly$^*$ continuous semigroup 
$\T=(\T_t)_{t\ge 0}$ of completely positive, identity preserving and
normal maps. We will make 
the assumption from now on that $\T$ is indeed uniformly  continuous 
i.e. $\lim_{t\to 0^+} \sup_{\Vert x\Vert \le 1}
\left\Vert \T_t(x)-x\right\Vert=0$. Its generator $\Ll$  can be then
represented in the well-known  (see \cite{Partha},\cite{SG})
 Gorini-Kossakowski-Sudarshan-Lindblad (GKSL)  form as 
\begin{equation}\label{eq:GKSL}
\Ll(x)=\mi [H,x]-\frac{1}{2}\sum_{\ell\geq 1}
\left(L_\ell^*L_\ell x-2L_\ell^*xL_\ell+xL_\ell^*L_\ell\right),
\end{equation}
where $H=H^*$ and $(L_\ell)_{\ell\geq 1}$ are operators on 
$\h$ such that the series $\sum_{\ell\geq 1}L^*_\ell L_\ell$ is
strongly convergent  and $[\cdot,\cdot]$ denotes the 
commutator $[x,y]=xy-yx$. 
The choice of operators $H$ and 
$(L_\ell)_{\ell\geq 1}$ is not unique (see Parthasarathy 
\cite{Partha} Theorem 30.16), however, this will not create 
any inconvenience in this paper.

Given a GKSL representation of $\Ll$ we call $\Ll_0$
\[
\Ll_0(x):=-\frac{1}{2}\sum_{\ell\geq 1}
\left(L_\ell^*L_\ell x-2L_\ell^*xL\ell+xL_\ell^*L_\ell\right),
\qquad x\in\B,
\]
dissipative part of $\Ll$ and $\mi \delta_H(x):=\mi[H,x]$  
Hamiltonian part  of $\Ll$ by abuse of language. Clearly, we 
have $\Ll=\mi\delta_H+\Ll_0$.

The \emph{decoherence-free (DF) subalgebra} of $\T$ is 
defined by
\begin{equation}\label{eq:NT-def}
\nT=\{x\in\B\,\mid\,\T_t(x^*x)=\T_t(x)^*\T_t(x),\ \T_t(xx^*)=\T_t(x)\T_t(x)^*\ \ \forall\,t\geq 0\}.
\end{equation}
It is a well known fact that $\nT$ is the biggest von Neumann subalgebra of $\B$ on which every $\T_t$ acts as a 
$*$-automorphism (see e.g. \cite{Evans} Theorem 3.1, \cite{DFR}  Proposition 2.1), {i.e. the system associated with $\nT$ is the biggest one evolving as a closed system. This explains the name given to this algebra.}

In the following proposition we recall some preliminary properties of $\nT$, whose proof can be found in \cite{DFSU}.

\begin{proposition}\label{prop-struct-NT}
Let $\T$ be a QMS on $\B$ and let 
${\mathcal{N}}(\T)$ be the set defined by \eqref{eq:NT-def}. 
Then
\begin{enumerate}
\item $\nT$ is $\T_t$-invariant for all $t\ge 0$,
\item for all $x\in\nT$, $y\in\B$ and $t\geq 0$ we have 
$\T_t(x^*y)=\T_t(x^*)\T_t(y)$ and 
$\T_t(y^*x)=\T_t(y^*)\T_t(x)$,
\item $\nT$ is a von Neumann subalgebra of $\B$,
\item $
\nT\subseteq \left\{x\in\B\,\mid\,\T_t(x)
=\hbox{\emph{e}}^{\mi tH}x\,\hbox{\emph{e}}^{-\mi tH}\ \  
\forall\,t\geq 0\right\}
$, for all self-adjoint operator $H$ in GKSL representation of $\Ll$,
\item $\nT$ is the commutant of the set 
of operators 
\begin{equation}\label{eq:iter-comm}
\left\{\delta_H^{n}(L_\ell),\delta_H^{n}(L_\ell^*)\,
\mid\,n\geq 0, \ell\geq 1\right\}.
\end{equation}
\end{enumerate}
\end{proposition}

In \cite{DFSU,FSU-NtFt} we have extensively studied the structure of $\nT$ when this algebra is atomic.\\

{We recall some preliminary definitions and results on von Neumann algebras.\\
Given a von Neumann algebra $\M$, we denote by $\Pmin{\M}$ the set of its minimal projections, and by $\Z{\M}$ its center, i.e. the von Neumann algebra
$$\Z{\M}:=\{x\in\M\,:[x,y]=0\ \ \forall\,y\in\M\}.$$
If $p$ is a projection in $\M$, its \emph{central support} $z_p$ is the smallest projection in $\Z{\M}$ such that $p\leq z_p$. 

A project $p$ is called \emph{finite} if, whenever we have $p=u^*u$ e $uu^*\leq p$ for some $u\in\M$, then $uu^*=p$ (Definition $6.3.1$ in \cite{kadison}).
\begin{definition}
Let $\M$ be a von Neumann algebra acting on $\h$. 
\begin{itemize}
\item[-]$\M$ is called \emph{atomic} if for every non-zero projection $p\in\M$ there exists $q\in\Pmin{\M}$, $q\neq0$, such that $q\leq p$.
%\item[-]$\M$ is of \emph{type I} if every non-zero central projection majorizes a non-zero abelian projection.
\item[-]$\M$ is a \emph{factor} if $\Z{\M}=\mathbb{C}\unit$.
\item[-]$\M$ is a \emph{type I factor} if it is a factor and possesses a non-zero minimal projection.
\item[-]$\M$ is a \emph{type II factor} if there are no minimal projections but there are non-zero finite projections. In particular $\M$ is a type $II_1$ factor if $\unit$ is finite.
\item[-]$\M$ is a \emph{type III factor} if it has no non-zero finite projections.
\end{itemize}
\end{definition}
We refer to Kadison book \cite{kadison} for these definitions (Definition $6.5.1$ and Corollary $6.5.3$). }

The following result gives some characterizations of the atomicity (see Theorem 5 of \cite{tomi3}, Theorem iv.2.2.2 of \cite{Black}, Theorem $6$ and Proposition A.1 in \cite{DFSU}).
\begin{proposition}\label{E-atomic}
Let $\M$ be a von Neumann subalgebra of $\B$. The following facts are equivalent:
\begin{enumerate}
\item $\M$ is atomic,
\item $\M$ is the image of a normal conditional expectation (i.e. a normal norm one projection) $\E:\B\to\M$,
\item there exists a countable family $(p_i)_{i\in I}$ of pairwise orthogonal minimal projections in $\Z{\M}$ such that $\sum_ip_i=\unit$ and $p_i\M p_i$ is a type I factor for all $i\in I$.
\end{enumerate}
\end{proposition}

In \cite{DFSU} we explained as the block-diagonal structure of an atomic decoherence-free subalgebra (see item $3$ of the previous proposition) has important consequences on the structure of the semigroup and its invariant states. In particular, the Lindblad operators in any GKSL representation of the generator inherit this block-decomposition.

More precisely our result is the following (see Theorem $3.2$ in \cite{DFSU}).

\begin{theorem}\label{th:main} 
$\nT$ is an atomic algebra if and only if there exist two countable sequences of Hilbert spaces $(\kk_i)_i$, $(\mm_i)_i$ such that (up to a unitary isomorphism) $\h=\oplus_{i\in I}\left(\kk_i\otimes\mm_i\right)$ and $\nT=\oplus_{i\in I}\left(\mathcal{B}(\kk_i)\otimes\unit_{\mm_i}\right)$.\\
In particular in this case, the following facts hold:
\begin{enumerate}
\item for every GKSL representation of $\Ll$ by means of operators 
$H,(L_\ell)_{\ell\ge 1}$, we have 
\[
L_\ell=\oplus_{i\in I}
\left(\unit_{\kk_i}\otimes N_\ell^{(i)}\right)
\] for a collection 
$(N_\ell^{(i)})_{\ell\geq 1}$ of operators in 
$\mathcal{B}(\mm_i)$, such that the series 
$\sum_{\ell\ge 1}N_\ell^{(i)*}N_\ell^{(i)}$  
{\red is} strongly convergent for all $i\in I$, and 
\[
H=\oplus_{i\in I}\left(H_i\otimes\unit_{\mm_i}
+\unit_{\kk_i}\otimes N_0^{(i)}\right)
\]
for  self-adjoint operators $N_i\in\mathcal{B}(\kk_i)$ 
and $N_0^{(i)}\in\mathcal{B}(\mm_i)$, $i\in I$,
\item  defining on the algebra
$\mathcal{B}\left(\oplus_{i\in I}
\left(\kk_i\otimes\mm_i\right)\right)$ 
\begin{equation}\label{eq:L-da-df}
\dfL = \mi\left[\oplus_{i\in I}(H_i\otimes\unit_{\mm_i}),\cdot\,\right]\end{equation}  and $\dissL$ as the Lindblad operator given by  $$\{\oplus_{i\in I}\left(\unit_{\kk_i}\otimes N_\ell^{(i)}\right), \oplus_{i\in I}(\unit_{\kk_i}\otimes N_0^{(i)})\,|\,l\geq 1\},$$
we find the commuting generators $\dfL$ and $\dissL$ of two 
commuting QMSs $\dfT$ (the decoherence-free semigroup) 
and $\dissT$ (the decoherence-affected semigroup)  such that 
$\T_t=\dissT_t\circ\dfT_t=\dfT_t\circ\dissT_t$. 
\item we have $\mathcal{N}(T^{\mm_i})=\mathbb{C}\unit_{\mm_i}$, where $\T^{\mm_i}$ is the QMS on $\mathcal{B}(\mm_i)$ whose GKSL generator is given by operators $\{N_0^{(i)}, N_\ell^{(i)}\}_\ell$.
\item the action of $\dfT$ is explicitly given by
$\dfT_t(x)=\hbox{\rm e}^{\mi t \widetilde{H} }x \hbox{\rm e}^{-\mi t  \widetilde{H} }$ for all $x\in\mathcal{B}\left(\oplus_{i\in I}
\left(\kk_i\otimes\mm_i\right)\right)$,
where $ \widetilde{H} $ is the self-adjoint operator $\oplus_{i\in I} (H_i\otimes\unit_{\mm_i})$; moreover
$\mathcal{N}(\dfT)= \mathcal{B}\left(\oplus_{i\in I}
\left(\kk_i\otimes\mm_i\right)\right)$ and
$\mathcal{N}(\dissT)=\nT$.
\item we have $\T_t(x\otimes y)={\rm e}^{\mi tH_i}x{\rm e}^{-\mi tH_i}\,\otimes\T^{\mm_i}(y)$ for $x\in\mathcal{B}(\kk_i)$, $y\in\mathcal{B}(\mm_i)$.
\end{enumerate}
\end{theorem}  
{Unfortunately, not all the decoherence-free subalgebras of uniformly continuous semigroups on $\B$ are atomic. In the following section we provide an example in which $\nT$ is a type $II_1$ factor.}

\subsection{Example of $\nT$ type $II_1$ factor} 
%In this section we provide an example of uniformly continuous QMS having as decoherence-free subalgebra a type $II_1$ factor.  
{First of all we recall some preliminary facts on the group theory (see for example section $6.7$ in \cite{kadison}).}

\smallskip
Given a discrete group $G$ with unit element $e$, we take $\h=\ell^2(G)$ with orthonormal basis $\{1_g\}_{g\in G}$, where $1_g(h)=\delta_{g,h}$.\\ Recalling that, for $u,v\in\h$ the convolution $u\ast v$ is defined as the element in $\ell^\infty(G)$ by
$$(u\ast v)(g)=\sum_{h\in G} u(g^{-1}h)\,v(h)=\sum_{h\in G}u^{-1}(h)\,v(hg),\qquad\forall\,g\in G,$$
we obtain two linear maps from $\h$ to $\ell^\infty(G)$ by setting
$$L_u(v)=u\ast v,\qquad R_u(v)=v\ast u.$$
In particular $L_{1_g}$ and $R_{1_g}$ belong to $\B$ for all $g\in G$, are unitary operators and they generate two von Neumann algebras, $\mathcal{L}_G$ and $\mathcal{R}_G$ respectively, such that $\mathcal{L}_G^\prime=\mathcal{R}_G$ (Theorem $6.7.2$ in \cite{kadison}). Moreover, the following facts hold:
\begin{enumerate}
\item $\mathcal{L}_G=\{L_u\,:\,u\in\h, L_u\in\B\}$ and $\mathcal{R}_G=\{L_u\,:\,u\in\h, R_u\in\B\}$,
\item $\left(L_{1_g}u\right)(h)=u(g^{-1}h)$ e $\left(R_{1_g}u\right)(h)=u(hg^{-1})$ for all $u\in\ell^2$,
\item $L_{1_g}+L_{1_h}=L_{1_g+1_h}$ and $\alpha L_{1_g}=L_{\alpha 1_g}$ for all $\alpha\in\mathbb{C}$,
\item $L_{1_g}L_{1_h}=L_{1_g\ast 1_h}=L_{1_{hg^{-1}}}$ and $L_{1_g}^*=L_{1_{g^{-1}}}$,
\item $L_{1_e}=\unit$.
\end{enumerate}
Similar results hold for $R_{1_g}$.\\

Now, Proposition $6.7.4$ and Theorem $6.7.5$ in \cite{kadison} give the following
\begin{theorem}
The von Neumann algebras $\mathcal{L}_G$ and $\mathcal{R}_G$ are finite. Moreover, if $G\neq\{e\}$ and the conjugacy class of every $g\neq e$ is infinite, then $\mathcal{L}_G$ and $\mathcal{R}_G$ are factors of type $II_1$.
\end{theorem}
An example of discrete group satisfying the condition of the theorem above is the free (non abelian) group $\mathcal{F}_n$ on $n$ generators (with $n\neq 2$), i.e. the group of \lq\lq words\,\rq\rq\ formed from the $n$ generators (see Example $6.7.6$ in \cite{kadison}).\\

We finally are in position to introduce the QMS having as $\nT$ a type II factor.

\begin{example}{\rm We consider $G=\mathcal{F}_n$ for some finite number $n\geq 2$, and $\h=\ell^2(G)$ as before. For every $x\in\B$ we consider the bounded operator
$$\Ll(x)=\sum_{g\in G}L_{1_{g}}^*xL_{1_{g}}-x=\sum_{g\in G}L_{1_{g^{-1}}}xL_{1_{g}}-x,$$
where the sum over $g\in G$ is countable being $G$ a discrete group.\\
Since $\Ll$ is expressed in a GKSL form (recall that operators $L_{1_g}$ are unitary), it generates a uniformly continuous QMS $\T$ on $\B$, and it satisfies
$$\nT=\{L_{1_g}\,:\,g\in G\}^\prime=\mathcal{L}_G^\prime=\mathcal{R}_G,$$
for the Hamiltonian is a multiple of the identity operator.\\
We can then conclude that $\nT$ is a type $II_1$ factor. Moreover, by Theorem \ref{invariant-atomic}, we will say that $\T$ has not faithful normal invariant states.
}
\end{example}

\section{Integral decomposition of $\nT$}
{We are now interested to understand what happens whenever $\nT$ is not atomic: can we also obtain some decomposition of this algebra? And, in this case, such a decomposition forces operators $\{H, L_k\}_k$ in any GKSL representation to have a suitable structure?

}
We will show that every element of $\nT$ is decomposable with respect to a suitable "disintegration" of $\h$ in direct integral of Hilbert spaces.
Moreover, this structure of the decoherence-free algebra induces a decomposition of Lindblad operators $\{H,L_k\}_k$ with respect to the same direct integral decomposition.
\medskip

To this end, we can apply Theorem \ref{th:disintegration} in Appendix to the von Neumann algebra $\mathcal{R}:=\Z{\nT}^\prime$  obtaining a decomposition of $\h$ in the direct integral of Hilbert spaces $(\h_\gamma)_\gamma$ over a (locally compact complete separable metric) measure space $(\Gamma,\mu)$,
$$ \h=\int_\Gamma^\oplus\h_\gamma\,\de\mu(\gamma).$$ 
Moreover there exists a family of von Neumann algebras $(\mathcal{R}_\gamma)_{\gamma\in\Gamma}$ on $(\h_\gamma)_{\gamma\in\Gamma}$ such that $\Z{\nT}^\prime$ has the integral decomposition
$$\Z{\nT}^\prime=\int_\Gamma^\oplus\mathcal{R}_\gamma\,\de\mu(\gamma),$$
i.e. every element of $\Z{\nT}^\prime$ is a decomposable operator.

Now, we want make use of the same theorem taking $\mathcal{A}:=\Z{\nT}$ as subalgebra of $\Z{\mathcal{R}}$, so that $\Z{\nT}$ is the diagonal algebra with respect this decomposition.

So, as a first step, we have to prove the inclusion $\Z{\nT}\subseteq\Z{\mathcal{R}}$.

Actually, by following lemma, we can say something more.

\begin{lemma}\label{centro-commutante}
If $\mathcal M$ is a von Neumann algebra, then
$$\mathcal Z(\mathcal M) = \mathcal Z(\mathcal Z(\mathcal M)^\prime).$$
\end{lemma}
\begin{proof}
The equality can be shown by a direct computation. Indeed, by the definition of center we have that
\(\mathcal Z(\mathcal Z(\mathcal M)') = \mathcal{Z}(\mathcal{M})'' \cap \mathcal{Z}(\mathcal{M})'\).
But the double commutant theorem gives \(\mathcal{Z}(\mathcal{M})''=\mathcal{Z}(\mathcal{M})\),
so that \(\mathcal Z(\mathcal Z(\mathcal M)')=\mathcal{Z}(\mathcal{M})\cap\mathcal{Z}(\mathcal{M})'\).
Finally, since the center of a von Neumann algebra is always abelian, we obtain \(\mathcal{Z}(\mathcal{M})\cap\mathcal{Z}(\mathcal{M})'=\mathcal{Z}(\mathcal{Z}(\mathcal{M}))=\Z{\M}\).
%\qedhere
\end{proof}

Therefore, if $\M=\nT$, the abelian algebra $\mathcal Z(\nT)$ coincides with $\mathcal Z(\Z{\nT}^\prime)=\Z{\mathcal{R}}$ and so, in particular, it is the maximal abelian subalgebra of $\mathcal Z(\Z{\nT}^\prime)$. By Theorem \ref{Mfactors} we have then the equality $\mathcal{R}_\gamma=\mathcal{B}(\h_\gamma)$ for almost every $\gamma\in\Gamma$.\\
These results are summed up in the following Proposition:
\begin{proposition}\label{NT-dec} There exists a (locally compact complete separable metric) measure space $(\Gamma,\mu)$ such that, up to unitary isomorphisms, $\h$ is the direct integral of Hilbert spaces $(\h_\gamma)_{\gamma\in\Gamma}$ over $(\Gamma,\mu)$ 
and $\mathcal Z(\nT)'$ is the algebra of all decomposable operators. In particular,
$$\Z{\nT}^\prime=\int_\Gamma^\oplus\mathcal{B}(\h_\gamma)\,\de\mu(\gamma).$$
Every element of $\nT$ is decomposable, i.e.
$$\nT=\int^\oplus_\Gamma \nT_\gamma\, \de\mu(\gamma)$$
where $\nT_\gamma$ are factors almost everywhere.\\
The center $\Z{\nT}$ of $\nT$ is the diagonal algebra with respect to this decomposition.
\end{proposition}
The decomposition of $\nT$ provided in proposition above is referred to as \lq\lq the central decomposition of $\nT$ (into factors)\rq\rq.

%$(\Gamma,\mu)$ is direct sum of a measure space $\Gamma_c$ without atoms (the \lq\lq continuous\rq\rq\ part of $\Gamma$), and disjoint discrete spaces $(\Gamma_n)_n$ endowed with the counting measures (see discussion after Theorem $9.4.1$ in \cite{kadison}). 
In this case there is no serious loss of generality if we think of $\Gamma$ as the unit interval, $\Gamma_c$, plus at most a countable number of atoms, $(\Gamma_n)_{n\in N}$. Therefore $\mu$ can be written as the sum of a \lq\lq continuous\rq\rq\ measure (Lebesgue measure on the unit interval) $\mu_{c}$ and discrete measures $(m_n)_{n\in N}$. So  we have that $$\h=\int_{\Gamma_c}^\oplus\h_\gamma\de\mu_c(\gamma)\oplus\left(\oplus_n\h_{\gamma_n}\right).$$ 
Now we can introduce some privileged projections. If we denote by $\h_{d}$ the discrete part of $\h$, i.e. $$\h_d:=\oplus_{m\in N}\h_{\gamma_m},$$ and by $\{e^m_{i}\}_{i\in I_m}$ an orthonormal basis of the (separable) Hilbert space $h_{\gamma_m}$, the family $\{e^m_{i}\}_{m\in N, i\in I_m}$ clearly gives an orthonormal basis of the $h_d$.
Defining orthogonal projections
\begin{equation}\label{projection-discret}
p_n=\kb{f_n}{f_n}\end{equation}
with $f_n=f_{n(m,i)}:=e^m_{i}$ for $n\in I:=\cup_{m\in N}I_m$,
and
\begin{equation}\label{projection-cont}
q:=\int_{{\Gamma_c}}^\oplus\unit_{\h_{\gamma}}\de\mu_c(\gamma),\end{equation}
we finally obtain a countable family $\{q, p_i\}_{i\in I}$ of mutually orthogonal and diagonal projections summing up to the identity, such that each $p_i$ is minimal in the center of $\nT$ and $p_i\nT p_i$ is a factor (see Proposition \ref{NT-dec}). 

Moreover, if $q\neq 0$, by the same proposition, $\Z{q\nT q}$ is the diagonal algebra with respect to the integral decomposition of $$\h_c:=\int_{\Gamma_c}^\oplus\h_\gamma\de\mu_c(\gamma).$$
Therefore, it is $*$-isomorphic to the multiplication algebra of $L^2(\Gamma_c,\mu_c)$ (see Examples $14.1.4(a)$ and $14.1.11(a)$), that does not possess minimal projections.

\begin{theorem}
The decoherence-free subalgebra $\nT$ can be decomposed as
\begin{equation}\label{NT-q-pi}
\nT=q\nT q\oplus\left(\oplus_{i\in I}\,p_i\nT p_i\right),
\end{equation}
where $\{q, p_i\}_{i\in I}$ is a countable family of mutually orthogonal projections in $\Z\nT$ summing {\red up} to the identity and satisfying the following properties:
\begin{enumerate}
\item $p_i\nT p_i$ is a factor,
\item $q\nT q$ either is zero or has diffuse center, i.e. without minimal projections.
\end{enumerate}
\end{theorem}

\subsection{Structure of the infinitesimal generator induced by the decoherence-free subalgebra}

In this section, we want to deduce a structure theorem for the infinitesimal generator, induced by the integral decomposition introduced in the previous section.

We recall a preliminary result (see Corollary $9.3.5$ in \cite{kadison}), in order to prove that every operator in the center of $\nT$ is a fixed point for the semigroup. This result was known in the case in which $\nT$ is atomic (Proposition 2.5 in \cite{DFSU}). Now we show that is true in a more generale framework.

\begin{lemma}\label{inner}
Let $\alpha:\M\to\M$ be a $*$-automorphism of type I von Neumann algebras. If $\alpha$ preserves $\Z{\M}$ then $\alpha$ is inner, that is there exists a unitary operator $U\in\M$ such that $\alpha(x)=UxU^*$ for all $x\in\M$.
\end{lemma}

\begin{proposition}\label{prop:znt}
The restriction of every $\T_t$ to $\mathcal Z(\nT)$ is a $*$-automorphism. In particular we have $\mathcal Z(\nT)\subseteq\FT$.
\end{proposition}

\begin{proof}
Since $\T_t$ acts a $*$-automorphism onto $\nT$, it is enough to show that its restriction to $\Z{\nT}$ is bijective.\\
So, let $x\in \mathcal Z(\nT)$ and $y\in\nT$; then there exists $z_t\in \nT$ such that $y=\T_t(z_t)$, and thus
$$\T_t(x)y=\T_t(x)\T_t(z_t)=\T_t(x z_t)=\T_t(z_tx)=y\T_t(x),$$
i.e. $\T_t(x)$ belongs to $\mathcal Z(\nT)$.

Viceversa, if $y\in \mathcal Z(\nT)$, in particular there exists $x\in \nT$ such that $\T_t(x)=y$, i.e $x=\mathrm{e}^{-itH}y\mathrm{e}^{itH}$ and so for every $z\in \nT$
$$zx=z\mathrm{e}^{-itH}y\mathrm{e}^{itH}=\mathrm{e}^{-itH}\T_t(z)y\mathrm{e}^{itH}=\mathrm{e}^{-itH}y\T_t(z)\mathrm{e}^{itH}=xz.$$
This means $x\in \mathcal Z(\nT)$.

In order to conclude the proof we have to show that every $x$ in $\Z{\nT}$ is a fixed point.\\
Since the restriction of $\T_t$ to $\Z{\nT}$ is a $*$-automorphism on a type I algebra coinciding with its center (being $\Z{\nT}$ commutative), Lemma \ref{inner} gives $\T_t(x)=UxU^*$ for all $x\in\Z{\nT}$, with $U$ a unitary operator in $\Z{\nT}$.
Therefore, the equality $\T_t(x)=x$ holds for all $x\in\Z{\nT}$.
\end{proof}

This result allows to provide the desired decomposition of Lindblad operators $\{H,L_k\}_k$ and it extends item 1 in Theorem \ref{th:main} in not atomic case.

\begin{theorem}
{In any GKSL representations of the generator $\Ll$ of $\T$}, the Lindblad operators $\{H,L_k\}_k$ are decomposable. More precisely, for almost every $\gamma\in \Gamma$ there exist $H_\gamma=H_\gamma^*$ and $(L_{k,\gamma})_k$ in $\mathcal B(\h_\gamma)$  such that
$$H=\int_\Gamma^\oplus H_\gamma\, \de\mu(\gamma)$$
$$L_k =\int_\Gamma^\oplus L_{k,\gamma}\, \de\mu(\gamma).$$
\end{theorem}
\begin{proof}
We know that $\nT$ is contained in the commutant of $L_k$ and $L_k^*$, so that $L_k$ and $L_k^*$ belong to $\nT^\prime$. Since $\nT$ clearly contains the diagonal algebra $\Z{\nT}$ (see Proposition \ref{NT-dec}), its commutant $\nT^\prime$ is also decomposable (see equation \eqref{int-commutant} in Appendix applied to $\mathcal{R}=\nT$ in Proposition $14.1.24$ of \cite{kadison}).
% (pg 636). 
On the other hand, by Proposition~\ref{prop:znt} the von Neumann algebra $\Z{\nT}$ is contained in the set of fixed points $\FT$, and so every projection of it commutes with the Lindbald operators $L_k$ and $H$. Consequently $H$ belongs to $\Z{\nT}^\prime$, since $\Z{\nT}$ is generated by its projections.
\end{proof}

For almost every $\gamma\in\Gamma$ we can then define on $\mathcal{B}(\h_\gamma)$ the uniformly continuous QMS $\T^\gamma$ whose generator $\Ll^\gamma$ is given by the Lindblad operators $\{H_\gamma, L_{k,\gamma}\}_k$.
In particular note that, for $x=\int_\Gamma^\oplus x(\gamma)\, \de\mu(\gamma)$ in $\nT$, $x(\gamma)\in\nT_\gamma$, we have
\begin{align*}\Ll(x)&=\mi\int_\Gamma^\oplus [H_\gamma,x(\gamma)] \,\de\mu(\gamma)\\
&-\frac{1}{2}\sum_k\int_\Gamma^\oplus\left( L_{k,\gamma}^*L_{k,\gamma}x(\gamma)-2L_{k,\gamma}^*x(\gamma)L_{k,\gamma} + x(\gamma)L_{k,\gamma}^*L_{k,\gamma}\right)\,\de\mu(\gamma)\\
&=\int_\Gamma^\oplus\Ll^\gamma(x(\gamma))\,\de\mu(\gamma).
\end{align*}
Now, recalling that $\Ll(x)$ is diagonalizable since it belongs to $\nT$, we get $\Ll(x)_\gamma=\Ll^\gamma(x(\gamma))$ for almost all $\gamma\in\Gamma$.\\
This means that 
\begin{equation}\label{T-gamma}
\T_t(x)_\gamma=\T^\gamma (x(\gamma))\qquad\forall\,x=\int_\Gamma^\oplus x(\gamma)\,\de\mu(\gamma)\in\nT
\end{equation}
and for almost every $\gamma$.

Thanks to this equation we immediately obtain the following result.
\begin{corollary}
Let be $\nT=\int_\Gamma^\oplus \nT_\gamma \,\de\mu(\gamma)$ the decomposition of $\nT$ in direct integrals. Then $$\nT_\gamma=\mathcal N(\T^\gamma)$$ for almost every $\gamma$, {where $\mathcal N(\T^\gamma)$ denotes the decoherence-free subalgebra of $\T^\gamma$.}
\end{corollary}

\begin{proof}
Let $x\in \nT$ with decomposition $x=\int_\Gamma^\oplus x(\gamma)\, \de\mu(\gamma)$. Since also $x^*x$ belongs to $\nT$ and $x^*x=\int_\Gamma^\oplus x(\gamma)^*x(\gamma)\,\de\mu(\gamma)$, equation \ref{T-gamma} and the equality $\T_t(x^*x)=\T_t(x^*)\T_t(x)$ give
$$\int_\Gamma^\oplus \T^\gamma(x^*(\gamma)x(\gamma))\,\de\mu(\gamma)=\int_\Gamma^\oplus \T^\gamma(x^*(\gamma))\T_t(x(\gamma))\,\de\mu(\gamma).$$
This means $\T^\gamma(x^*(\gamma)x(\gamma))=\T^\gamma(x^*(\gamma))\T_t(x(\gamma))$ for almost every $\gamma$. In a similar way we obtain $\T^\gamma(x(\gamma)x^*(\gamma))=\T^\gamma(x(\gamma))\T_t(x^*(\gamma))$ for almost every $\gamma$, and so $x(\gamma)\in \mathcal N(\T^\gamma)$ for almost every $\gamma$. This prove the inclusion $\nT_\gamma\subseteq \mathcal N(\T^\gamma)$. 

Viceversa it is trivial. \end{proof}

%\section{The existence of a faithful invariant state}
\section{$\nT$ atomic}
We recall that the decoherence-free subalgebra $\nT$ is strongly related to the property of environmental decoherence. (see e.g. \cite{Hellmich,CSU2,CSU-new,DFSU,FSU-NtFt}).
Indeed, given $\T$ a uniformly continuous QMS on $\B$, there is \emph{environment induced decoherence}
(EID) on the open system described by $\T$ if there exists a $\T_t$-invariant and
$*$-invariant weakly$^*$ closed subspace $\M_2$ of $\B$ such that:
 \begin{itemize}
 \item[(EID1)] $\B=\nT\oplus\M_2$ with $\M_2\not=\{0\}$,
 \item[(EID2)] $w^*-\lim_{t\to\infty}\T_t(x)=0$ for all $x\in\M_2$.
 \end{itemize}
Unfortunately, if $\h$ is infinite-dimensional, it is not clear when such a decomposition exists and, in the case, if the space $\M_2$ is uniquely determined. However,
 $\mathcal{M}_2$ is always contained in the $\T_t$-invariant and
$*$-invariant closed subspace
$$
\M_0=\left\{\, x\in\B\,:\,w^*-\lim_{t\to\infty}\T_t(x)=0\,\right\}.
$$
\smallskip

Since the decomposition $\B=\nT\oplus\M_2$ is clearly related to the asymptotic properties of the semigroup, it is very natural to compare it with another famous asymptotic decomposition of $\B$, called the \emph{Jacobs-de Leeuw-Gliksberg splitting}.

We recall that, when \underline{there exists a faithful normal invariant state} $\rho$, the Jacobs-de Leeuw-Gliksberg
splitting holds (see e.g. Corollary $3.3$ and Proposition $3.3$ in \cite{Hellmich}) giving $\B=\mathfrak{M}_r\oplus\mathfrak{M}_s$ with
\begin{align}
\mathfrak{M}_r&:=\overline{\spanno}^{w^*}\{x\in\B\,:\,
\T_t(x)=e^{\mi t\lambda}x\ \mbox{for some $\lambda\in\mathbb{R}$},\ \forall\,t\geq 0\}\\
\mathfrak{M}_s&:=\{x\in\B\,:\,0\in\overline{\{\T_t(x)\}}^{w^*}_{t\geq 0}\}.
\end{align}
%$\mathfrak{M}_r$ is called the \emph{reversible algebra}, while $\mathfrak{M}_s$ is the \emph{stable space}.\\
Moreover, in this case $\mathfrak{M}_r$ is a von Neumann subalgebra (called \emph{reversible algebra}), the action of each $\T_t$ on it is a $*$-automorphism, and there exists a normal conditional expectation onto $\mathfrak{M}_r$ compatible with $\rho$. In particular this means that $\mathfrak{M}_r$ is atomic by Proposition \ref{E-atomic}, and it is contained in $\nT$.\\ Therefore, it is natural to ask us if we can have the equality $\mathfrak{M}_r=\nT$. 

A first answer to this problem is given in \cite{FSU-NtFt}, Theorem $11$, when $\nT$ is atomic. 
\begin{theorem}\label{equiv-NT-atom}
Assume there exists a faithful normal invariant state. Then $\nT$ is atomic
if and only if EID holds with $\nT=\mathfrak{M}_r$ and $\M_2=\M_0$.
\end{theorem}
This means that, in this case, the decomposition induced by decoherence is unique, i.e. the only way to realize it, is taking $\nT=\mathfrak{M}_r$ (and, consequently, $\M_2=\M_0$).

In this section we want to show that, \emph{if $\T$ has a faithful invariant state, the decoherence-free subalgebra $\nT$ is always atomic and it coincides with the reversible algebra}. At last this result allows us to provide an answer to the initial problem of the existence and uniqueness of the decomposition induced by decoherence.\\

Recalling {the decomposition given in Equation \eqref{NT-q-pi},
$$\nT=q\nT q\oplus \bigoplus_i p_i\nT p_i,$$
the first step is to prove that the projection $q$ is zero}.

\begin{proposition}\label{ntfactors} If there exists a normal faithful invariant state then $\nT$ is direct sum of factors.
\end{proposition}
\begin{proof} It is known that if there exists a normal faithful invariant state $\FT$ is atomic, since it is the image of a normal conditional expectation. By Proposition \ref{prop:znt} $\mathcal Z(\nT)\subseteq\FT$, so $\mathcal Z(\nT)\subseteq \mathcal Z(\FT)$. If we assume by contradiction that {$q$ is not a zero projection,
for every measurable set $A\subseteq\Gamma_c$ with $\mu_c(A)>0$ the operator}
$$p_A=\int_A \unit_{\h_\gamma} d\mu_{c}(\gamma)$$
is a diagonal projection, and so it belongs to $\mathcal Z(\nT)\subseteq \mathcal F(\T)$. Thus we have projection a $p_A$ in $\FT$ that does not majorize a minimal projection, and this is a contradiction for the atomicity of $\FT$. 
%So $\nT=\oplus_i p_i\nT p_i$, where $p_i\nT p_i$ are factors and $p_i$ are minimal projections in the center of $\nT$ for every $i$. 
\end{proof}

Therefore, if there exists a faithful normal invariant state, Proposition before gives the decomposition
\begin{equation}\label{dec-NT-fattori}\nT=\bigoplus_i p_i\nT p_i\end{equation}
for a suitable family of mutually orthogonal projections $(p_i)_i$ in $\nT$ summing up to the identity and minimal in the center of $\nT$. Moreover every $p_i\nT p_i$ is a factor. Our aim is now to prove that each of them is a type I factor.

First of all note that, since each $p_i$ is a fixed point, $p_i\nT p_i$ is the decoherence-free subalgebra of the QMS $\T^i$ on $\mathcal{B}(p_i\h)$ given by the restriction of the semigroup to this algebra. In other words we have $p_i\nT p_i=\mathcal N(\T^{i})$. If we prove that also every $\T^{i}$ possesses a faithful normal invariant state, then we can reduce to study QMSs having a factor as decoherence-free subalgebra.\\
As a first step we show that invariant states inherit the block structure of $\nT$.

\begin{lemma}\label{lem-fuori-diag} {Assume there exists a faithful normal invariant state.} Let $(p_i)_i$ be the family of projections in \eqref{dec-NT-fattori} and $\sigma$ be an invariant state. Then $p_i\sigma p_j=0$, whenever $i\neq j$.\end{lemma}
\begin{proof} Since central projections $p_i,p_j$ are in 
$\mathcal{F}(\T)$, for all $t\ge 0$ and $x\in\B$ we have 
$\T_t(p_i x p_j) = p_i \T_t(x) p_j$ and also 
$\T_{*t}(p_i \sigma p_j) = p_i\T_{*t}(\sigma) p_j$ for 
all trace class operator $\sigma$. It follows that 
\[
\tr{p_i \sigma p_j x} 
=\tr{p_i \T_{*s}(\sigma ) p_j x}
=\tr{ \sigma  \T_{s}(p_i x p_j )}
=\tr{ \sigma  p_i \T_{s}(x ) p_j}
\]
for all invariant state $\sigma$ and $x\in\B$.
Integrating on $[0,t]$ and dividing by $t$ we find
\[
\tr{p_i\sigma p_j x}  = \tr{\sigma p_i
\left(t^{-1}\int_0^t\T_{s}(x ) ds\right) p_j},
\]
and, taking the limit as $t\to\infty$, we have, 
\[
\tr{p_i\sigma p_j x} = \tr{p_i\sigma p_j\mathcal{E}(x)}
\]
where $\mathcal{E}$ is the conditional expectation onto $\mathcal{F}(\T)$, thus $\mathcal E(x)$ is a decomposable operator. Now, since $\FT$ 
is also contained in $\nT=\oplus_{k\in I}\, p_k\nT p_k$,  we get 
$p_i \mathcal{E}(x)p_j=0$ for $i\neq j$ as well as  
$\tr{p_i\sigma p_j x}=0$, and so $p_i\sigma p_j =0$ by the 
arbitrarity of $x$.
\end{proof}

As a consequence of the previous lemma, every faithful normal invariant state is decomposable as a direct sum of faithful invariant functionals. More precisely:
\begin{proposition}\label{struct-inv-states}{Assume there exists a faithful normal invariant state.}
If $\sigma$ is a normal invariant state, then $\sigma=\sum_{i\in I} p_i \sigma p_i$ where every $p_i \sigma p_i$ is a (eventually zero) normal invariant functional for $\T^{i}$.\\
Moreover, if $\sigma$ is faithful, then $$\sigma_i:=\frac{p_i \sigma p_i}{\tr{p_i\sigma p_i}}$$ is a faithful state on $\mathcal{B}(p_i\h)$.\\
\end{proposition}

Therefore, as we said before, we can assume that $\T$ is a QMS on $\B$ such that its decoherence-free subalgebra $\nT$ is a factor.

\smallskip
 In order to prove that $\nT$ is a type I factor, we investigate its relation with $\mathfrak{M}_r$. In the following we will denote this algebra by $\mathcal{R}(\T)$ in such a way that it will be clear to which semigroup the algebra refers.\\
Since $\Rr\T$ is atomic, we can provide a block decomposition of it through a simple change of the proof of Theorem \ref{th:main} (see Theorem 3.2 in \cite{DFSU}). 
 
%As in \cite{NT-FT} we proved that when $\nT$ is atomic it induces a decomposition of the infinitesimal generator (up to an isometric isomorphism), {red we can provide the same result for $\Rr{\T}$ through a simple change o}f the proof of Theorem 3.2 in \cite{NT-RMP} and Theorem 5 in \cite{NT-FT}. 

\begin{lemma}\label{lem-Mr}
If there exists a faithful normal invariant state, the following facts hold:
\begin{enumerate}
\item there exist two sequences of separable Hilbert spaces $(\Ss_j)_{j\in J}$ and $(\ff_j)_{j\in J}$ such that $\h=\oplus_{j\in J}\left(\Ss_j\otimes\ff_j\right)$ and
$$\Rr\T=\oplus_{j\in J}\left(\mathcal{B}(\Ss_j)\otimes\unit_{\ff_j}\right),$$
\item for every GSKL representation of $\Ll$ by means of operators
$H,(L_\ell)_{\ell\ge 1}$, we have
\[
L_\ell =\oplus_{j\in J}
\left(\unit_{\Ss_j}\otimes M_\ell^{(j)}\right)
\]
for a collection $(M_\ell^{(j)})_{\ell\geq 1}$ of operators in $\mathcal{B}(\ff_j)$,
such that the series $\sum_{\ell\ge 1}M_\ell^{(j)*}M_\ell^{(j)}$
strongly convergent for all $j\in J$, and
\[
H=\oplus_{j\in J}\left(K_j\otimes\unit_{\ff_j}
+\unit_{\Ss_j}\otimes M_0^{(j)}\right)
\]
for  self-adjoint operators $K_j\in\mathcal{B}(\Ss_j)$
and $M_0^{(j)}\in\mathcal{B}(\ff_j)$, $j\in J$,
\item defining on $\mathcal{B}(\ff_j)$ the GKSL generator $\Ll^{\ff_j}$ associated with operators $\{M_0^{(j)}, M_\ell^{(j)})\,:\,\ell^{(j)}\geq 1\}$, we have $$\T_t(x_j\otimes y_j)=\hbox{\rm e}^{\mi t K_j }x_j\hbox{\rm e}^{-\mi t K_j }\otimes\T^{\ff_j}(y_j)$$
for all $t\geq 0$, $x_j\in\mathcal{B}(\Ss_j)$ and $y_j\in\mathcal{B}(\ff_j)$, where $\T^{\ff_j}$ is the QMS generated by $\Ll^{\ff_j}$, 
\item $\Rr{\T^{\Ss_j}}=\mathcal{N}(T^{\Ss_j})=\mathcal{B}(\Ss_j)$ and $\Rr{\T^{\ff_j}}=\C\unit_{\ff_j}$ for all $j\in J$, where $T^{\Ss_j}$ denotes the QMS on $\mathcal{B}(\Ss_j)$ generated by $\Ll^{\Ss_j}=\mi[K_j,\,\cdot\,]$,
\item $K_j$ has pure point spectrum for all $j\in J$.
\end{enumerate}
\end{lemma}

In the following we will use notations of Lemma \ref{lem-Mr}.
\begin{theorem} \label{invariant-factor}
Assume there exists a normal faithful invariant state and $\Rr{\T}$ is a factor. Then $\nT=\Rr{\T}$.
\end{theorem}

\begin{proof}
First of all note that, since $\Rr{\T}$ is atomic, if it is a factor, it has to be a type I factor. Then assume $\Rr{\T}=\mathcal{B}(\Ss)\otimes\unit_\ff$ for some separable Hilbert spaces $\Ss$ and $\ff$ such that $\h=\Ss\otimes\ff$.

We claim that $\mathcal{N}(\T^\ff)=\Rr{\T^\ff}$.\\
By item $3$ the algebra $\Rr{\T^\ff}$ is trivial, i.e. $\Rr{\T^\ff}=\C\unit_\ff$. So, since $\mathcal{N}(\T^\ff)$ contains $\Rr{\T^\ff}$, we assume there exists a non-zero projection $p$ in $\mathcal{N}(\T^\ff)$ and prove that $p=\unit_\ff$.\\
The Jacobs-de-Leeuw-Glicksberg splitting of $\mathcal{B}(\ff)$ gives $\mathcal{B}(\ff)=\C\unit_\ff\oplus\mathfrak{F}_s$, with $$\mathfrak{F}_s:=\{x\in\mathcal{B}(\ff)\,:\,0\in\overline{\{\T^\ff_t(x)\}}^{w^*}_{t\geq 0}\}$$
the corresponding stable space. Therefore $p=\mu\unit_\ff+y$ for some $\mu\in\C$, $y\in\mathfrak{F}_s$, and $w^*-\lim_\alpha\T^\ff_{t_\alpha}(y)=0$ for a generalized sequence $(t_\alpha)_\alpha\subseteq\R^+$ going to infinity. This implies there exists $w^*-\lim_\alpha\T^\ff_{t_\alpha}(p)=\mu\unit_\ff$.\\
On the other hand, since $\T^\ff$ is a uniformly continuous QMS on $\mathcal{B}(\ff)$, its action on the decoherence-free subalgebra is unitary and defined by the Hamiltonian $M_0$, i.e. $\T_t^\ff(x)=\mathrm{e}^{\mi tM_0}x\mathrm{e}^{-\mi tM_0}$ for all $x\in\mathcal{N}(\T^\ff)$. In particular $$\T_t^\ff(p)={\rm}\mathrm{e}^{\mi tM_0}p\mathrm{e}^{-\mi tM_0}\quad\forall\,t\geq 0.$$ Therefore, denoting by $(e_n)_{n\in\mathbb{N}}$ an orthonormal basis of $\ff$ given by eigenvectors of $M_0$ (see item $4$ in Lemma \ref{lem-Mr}), $M_0e_n=\lambda_ne_n$ with $\lambda_n\in\R$, and taking the normal functional $\sigma:=\ee{k}{l}$ for $l,k\in\mathbb{N}$, we obtain
\begin{align*}\mu\delta_{lk}&=\tr{\sigma\mu\unit_\ff}=\lim_\alpha\tr{\sigma\T^\ff_{t_\alpha}(p)}=\lim_\alpha\scal{e_l}{\mathrm{e}^{\mi t_\alpha M_0}p\mathrm{e}^{-\mi t_\alpha M_0}e_k}\\
&=\lim_\alpha \mathrm{e}^{\mi t_\alpha(\lambda_l-\lambda_k)}\scal{e_l}{pe_k}.
\end{align*}
This means that $\scal{e_l}{pe_l}=\mu$ for all $l$, and $\scal{e_l}{pe_k}=0$ for $l\neq k$, i.e. $p=\mu\unit_\ff$, and the equality $\mathcal{N}(\T^\ff)=\C\unit_\ff=\Rr{\T^\ff}$ is then proved, being $p$ a projection.

Now we show that $\nT=\mathcal{N}(\T^\Ss)\otimes\mathcal{N}(\T^\ff)=\mathcal{B}(\Ss)\otimes\unit_\ff$. \\
So, let $x\in \mathcal{B}(\Ss)=\mathcal{N}(\T^\Ss)$ and $y\in \mathcal{B}(\ff)$ such that $x\otimes y\in\nT$. Since {$x\otimes y^*y=(\unit_\Ss\otimes y^*)(x\otimes y)$, by item $2$ in Proposition \ref{prop-struct-NT} we have
\begin{align*}
{\rm}\mathrm{e}^{\mi tK}x\mathrm{e}^{-\mi tK}\otimes\T^\ff(y^*y)&=\T_t(x\otimes y^*y)=\T_t(\unit_\Ss\otimes y^*)\T_t(x\otimes y)\\
&=\left(\unit_\Ss\otimes\T_t^\ff(y^*)\right)\left({\rm}\mathrm{e}^{\mi tK}x\mathrm{e}^{-\mi tK}\otimes\T_t^\ff(y)\right)\\
&={\rm}\mathrm{e}^{\mi tK}x\mathrm{e}^{-\mi tK}\otimes\T_t^\ff(y^*)\T_t^\ff(y)
\end{align*}
for all $t\geq 0$,
}
giving $\T^\ff(y^*y)=\T_t^\ff(y^*)\T_t^\ff(y^*)$. Analogously, $\T^\ff(yy^*)=\T_t^\ff(y)\T_t^\ff(y^*)$, so that $y$ belongs to $\mathcal{N}(\T^\ff)$. 

On the other hand, the inclusion $\mathcal{B}(\Ss)\otimes\unit_\ff\subseteq\nT$ is clear thanks to the structure Theorem (Lemma \ref{lem-Mr}).
\end{proof}

We can now to show that $\nT$ is a type I factor.

\begin{theorem}\label{invariant-I-factor}
If there exists a normal faithful invariant state and $\nT$ is a factor, then $\nT$ is a type I factor. In particular, $\nT$ coincides with the reversible algebra $\Rr\T$.
\end{theorem}
\begin{proof}
Let
$$
\Rr\T=\oplus_j \left(q_j\Rr\T q_j\right)=\oplus_j \left(\mathcal B(\Ss_j)\otimes\unit_{\ff_j}\right)
$$
    be the atomic decomposition of $\Rr\T$, where $q_j$ are the projections onto $\Ss_j\otimes \ff_j\subseteq\h$. Since each $q_j$ is a fixed point, the algebra $q_j\B q_j$ is preserved by the action of every map $\T_t$ and the restriction of the semigroup to this algebra is a QMS $\T^{(j)}$ on $\mathcal B(q_j\h)$ satisfying $\mathcal N(\T^{(j)})=q_j\mathcal N(\T)q_j$. Moreover Theorem \ref{invariant-factor} gives $$q_j\mathcal N(\T)q_j=\mathcal B(\Ss_j)\otimes\unit_{\ff_j}=\Rr{\T^{(j)}}.$$ So by choosing $(e^{(j)}_n)_n$ an onb of $\Ss_j$, the operator $q:=\kb{{\rm e}^{\mi(j)}_1}{e^{(j)}_1}\otimes\unit_{\ff_j}$ is a minimal projection in $q_j\nT q_j$ which is also minimal in $\nT$. Indeed, if there exists another projection $q'\in \nT$ such that $q'\leq q$, we have $\Ran(q')\subseteq\Ran(q)\subseteq q_j\h$, and so $q'$ belongs to $\mathcal N(\T^{(j)})$ giving either $q'=q$ or $q'=0$.\\
Therefore, since $\nT$ is a factor possessing a non zero minimal projection, it is a type $I$ factor and then $\nT=q_j\nT q_j=\mathcal B(\Ss_j)\otimes\unit_{\ff_j}$ for a unique $j\in J$. Finally, the atomic decomposition of $\Rr\T\subseteq\nT$ implies $\Rr\T=\mathcal B(\Ss_j)\otimes\unit_{\ff_j}=\nT$.
\end{proof}

We are finally in position to provide the main result of the paper.
\begin{corollary}\label{invariant-atomic}
If there exists a normal faithful invariant state then $\nT$ is atomic. In particular, $\nT$ coincides with the reversible algebra $\Rr\T$.
\end{corollary}
\begin{proof} Let $\nT=\bigoplus_{i\in I} p_i\nT p_i$ be the decomposition of $\nT$ given by Proposition \ref{dec-NT-fattori}. Since each factor $p_i\nT p_i$ coincides with the decoherence-free subalgebra $\mathcal N(\T^{i})$ of $\T^i$, and $\T^i$ has a faithful normal invariant state (see Proposition \ref{struct-inv-states}), Theorem \ref{invariant-I-factor} shows that $p_i\nT p_i=\Rr{\T^i}$ and it is a type I factor. Therefore $\nT=\oplus_{i\in I} \Rr{\T^i}=\Rr\T$ and it is atomic.
\end{proof}

{Corollary above totally solves the problem of decoherence for uniformly continuous QMSs having a faithful normal invariant state. Indeed, it provides a unique decomposition of the algebra $\B$ and shows the atomicity of $\nT$, generalizing what happens in the finite-dimensional case.

Note that in \cite{OL}, the authors had already achieved the same decomposition induced by decoherence when the semigroup commutes with the modular group associated with the faithful invariant state. Here we do not need this condition (that ensures the atomicity of $\nT$). However, for a not necessarily uniformly continuous QMS satisfying other additional hypothesis, they proved EID with respect to a normal semifinite and faithfull weight.}

\begin{remark} {\rm In \cite{CJ}, the authors proved that for a quantum channel acting on $\B$ and with a faithful invariant state, the reversible subalgebra coincides with the decoherence-free subalgebra. The previous corollary reaches the same conclusion also in the continuous setting but through different techniques.}\end{remark}
\appendix
\section{Direct Integrals of von Neumann Algebras}
In this appendix we briefly recall some results about the general theory of the direct integrals of von Neumann algebras. {Our} aim is to decompose {them} into factors and generalize the concept of direct sum of von Neumann algebras. We follow the notations of \cite{kadison}.\\

\begin{definition}[$14.1.1$ in \cite{kadison}]
Let \(\Gamma \) be a Borel space ($\sigma-$compact locally compact space) with a \(\sigma \)-finite Borel measure \(\mu \) and  $(\h_\gamma)_{\gamma\in \Gamma}$ be a family of non-zero separable Hilbert spaces indexed by points $\gamma$ of $\Gamma$. We say that a separable Hilbert space $\h$ is the {\bf direct integral of $(\h_\gamma)_{\gamma\in \Gamma}$ over $(\Gamma,d\mu)$}, i.e.
$$\h=\int_\Gamma^\oplus \h_\gamma d\mu(\gamma)$$
when to each $u\in \h$ then corresponds a function $\gamma\to u(\gamma)$ on $\Gamma$ such that $u(\gamma)\in\h_\gamma$ for each $\gamma$ and 
\begin{itemize}
\item $\gamma\to \scal{u(\gamma)}{v(\gamma)}$ is $\mu-$integrable, when $u,v\in\h$ and 
$$\scal{u}{v}=\int_\Gamma \scal{u(\gamma)}{v(\gamma)}d\mu(\gamma),$$
\item if $u_\gamma\in\h_\gamma$ for all $\gamma$ in $\Gamma$ and $\gamma\to \scal{u_\gamma}{v(\gamma)}$ is integrable for each $v\in \h$, then there is a $u\in\h$ such that $u(\gamma)=u_\gamma$ for almost every $\gamma$. \end{itemize}
We say that $\int_\Gamma^\oplus \h_\gamma d\mu(\gamma)$ and $\gamma\to u(\gamma)$ are the {\bf (direct integral) decompositions} of $\h$ and {\blu $u$}, respectively.
\end{definition}

\begin{example}\rm
\begin{enumerate} 
\item If we consider a constant family, $\h_\gamma=\h$ for all $\gamma\in\Gamma$, then the direct integral $\h=\int_{\Gamma}^{\oplus}\h_\gamma\,\de\mu(\gamma)$ is just the space of measurable functions from $\Gamma$ to $\h$ which are square-integrable with respect to $\mu$, that is $\h=L^2(\Gamma,\mu;\h)$.
\item If $\Gamma$ is discrete and $\mu$ is counting measure on $\Gamma$, then
$$\h=\int_{\Gamma}^{\oplus}\h_\gamma\,\de\mu(\gamma)=\oplus_{\gamma\in\Gamma}\h_\gamma.$$
\end{enumerate}
\end{example}

\begin{definition}[$14.1.6$ in \cite{kadison}] If $\h$ is the direct integral of $(\h_\gamma)_{{\gamma\in\Gamma}} $ over $(\Gamma,\mu)$, an operator $x$ in $\B$ is said to be {\bf decomposable} when there is a function $\gamma\to x(\gamma)$ on $\Gamma$ such that $x(\gamma)\in\mathcal B(\h_\gamma)$ and, for each $u$ in $\h$, $x(\gamma)u(\gamma)=(xu)(\gamma)$ for almost every $\gamma$. If in addition $x(\gamma)=f(\gamma)\unit_{\h_\gamma}$ for some $f\in L^\infty(\Gamma,\mu)$, we say that $x$ is {\bf diagonalizable}.
\end{definition}

\begin{remark}\label{Mf}\rm
Let $f\in L^\infty(\Gamma,\mu)$. Then for every $\xi\in\h$ the map $\gamma\mapsto f(\gamma)\xi(\gamma)\in\h_\gamma$ is measurable and $\nm f(\gamma)\unit_\gamma\nm\leq\nm f\nm_\infty$ i.e. the function $\gamma\mapsto\nm f(\gamma)\unit_\gamma\nm$ is in $L^\infty(\Gamma,\mu)$. This means that the operator $M_f$ defined by $M_f(\gamma)=f(\gamma)\unit_\gamma$ is diagonal.

Viceversa, every diagonal operator $x$ has the form $x=M_f$ for some $f\in L^\infty(\Gamma,\mu)$. Indeed, assuming $x$ positive without loss of generality, by definition, we have $0\leq x(\gamma)=f(\gamma)\unit_\gamma$ almost everywhere, with $f(\gamma)\geq 0$. Therefore $|f(\gamma)|=\nm x(\gamma)\nm\leq\nm x\nm$ giving $f\in L^\infty(\Gamma,\mu)$.

In particular, if $f$ is the characteristic function of some measurable set $A$, then $M_f$ is a projection (i.e. the diagonalizable projection corresponding to $A$) and it can be written as $p=\int_A \unit_{\h_\gamma}d\mu(\gamma)$.
\end{remark}

\begin{proposition} If $\h$ is the direct integral of $(\h_\gamma)_{\gamma\in\Gamma}$ over $(\Gamma,\mu)$ and $x_1,x_2$ are decomposable, sep-adjoint operators on $\h$ such that $x_1\leq x_2$, then $x_1(\gamma)\leq x_2(\gamma)$ almost everywhere. If $x$ is decomposable, then $\gamma\to \|x(\gamma)\|$ is an essentially bounded measurable function with essential bound $\|x\|$.
\end{proposition}

Now we can introduce the concept of decomposable von Neumann subalgebra $\mathcal{R}$ acting on $\h$, whenever $\mathcal{R}$ is a subalgebra of the algebra of decomposable operators. Moreover, the decomposition $\gamma\to\mathcal{R}_\gamma$ is unique, so we write
\begin{equation}
	\mathcal{R}=\int_\Gamma^\oplus\mathcal{R}_\gamma\,\de\mu(\gamma)
\end{equation}
Moreover we have that if $\mathcal{R}$ contains the algebra of diagonalizable operators then $(\mathcal{R}')(\gamma)=(\mathcal{R}_\gamma)'$ almost everywhere (Proposition 14.1.24 \cite{kadison}), {i.e.
\begin{equation}\label{int-commutant}
\mathcal{R}^\prime=\int_\Gamma^\oplus\mathcal{R}_\gamma^\prime\,\de\mu(\gamma).
\end{equation}}
Finally, the center \(Z(\mathcal{R})\) of $\mathcal{R}$ is also expressed as a direct integral
\begin{equation}\label{int-center}
	Z(\mathcal{R}) = \int_\Gamma^\oplus Z(\mathcal{R})_\gamma\,\de\mu(\gamma)
\end{equation}
where $Z(\mathcal{R})_\gamma$ coincides with $Z(\mathcal{R}_\gamma)$.
In particular, \(Z(\mathcal{R})\) coincides with the diagonal algebra if and only if
\(\mathcal{R}_\gamma\) is a factor for almost every \(\gamma\in\Gamma \).

Now we want to apply this theory to prove that, given a von Neumann algebra $\mathcal{R}$ on a separable Hilbert space $\h$, any von Neumann subalgebra \(\mathcal{A}\) of the the center \(Z(\mathcal{R})\) induces a integral decomposition of $\h$ such that $\mathcal{R}$ becomes decomposable. 

\begin{theorem}\label{th:disintegration}[Theorem $14.2.1$, Theorem $14.2.2$ {in} \cite{kadison}]
{Let $\mathcal{R}$ be a von Neumann algebra acting on a separable Hilbert space \(\h \). 
If $\mathcal A$ is an abelian von Neumann subalgebra of the center of $\mathcal{R}$, then:}
\begin{enumerate}
\item there is a (locally compact complete separable metric) measure space $(\Gamma,\mu)$ such that $\h$ is (unitarily equivalent to) the direct integral of Hilbert spaces $(\h_\gamma)_{\gamma\in\Gamma}$ over $(\Gamma,\mu)$,
\item $\mathcal A$ is (unitarily equivalent to) the algebra of diagonalizable operators relative to this decomposition,
\item $\mathcal{R}$ is a decomposable von Neumann subalgebra, {i.e.
$$\int_\Gamma^\oplus \mathcal{R}_\gamma\,\de\mu(\gamma)$$
for a suitable family of von Neumann algebras $(\mathcal{R}_\gamma)_{\gamma\in\Gamma}$ acting on $\h_\gamma$.}
In particular $\mathcal{R}_\gamma$ is a factor for almost every $\gamma$ if and only if $\mathcal A$ coincides with $\mathcal Z(\mathcal{R})$.
\end{enumerate}
\end{theorem}

\begin{theorem}\label{Mfactors}[Theorem $14.2.4$ in \cite{kadison}] If $\mathcal{R}$ is a von Neumann algebra acting on a separable Hilbert space $\h$ and $\h$ is the direct integral of $(\h_\gamma)_{\gamma\in\Gamma}$ in a decomposition relative to an abelian von Neuman subalgebra of $\mathcal A$ of $\mathcal{R}'$, then $\mathcal{R}_\gamma=\mathcal B(\h_\gamma)$ almost everywhere if and only if $\mathcal A$ is a maximal abelian subalgebra of $\mathcal{R}'$.
\end{theorem}

{
\section*{Acknowledgments}
The financial support of MIUR FIRB 2010 project RBFR10COAQ
{\it Quantum Markov Semigroups and their Empirical Estimation}
and of GNAMPA 2020, project PRR-20200121-130335-957
{\it Evoluzioni Markoviane Quantistiche} are gratefully acknowledged.
}


\begin{thebibliography}{10}

\bibitem{LW}
D.~A. Lidar and B.~K. Whaley.
\newblock {\em Decoherence-Free Subspaces and Subsystems}, pages 83--120.
\newblock Springer Berlin Heidelberg, Berlin, Heidelberg, 2003.

\bibitem{KL}
E.~Knill, R.~Laflamme, and L.~Viola.
\newblock Theory of quantum error correction for general noise.
\newblock {\em Phys. Rev. Lett.}, 84:2525--2528, Mar 2000.

\bibitem{BlOl}
Ph. Blanchard and R.~Olkiewicz.
\newblock Decoherence induced transition from quantum to classical dynamics.
\newblock {\em Reviews in Mathematical Physics}, \textbf{15}(03):217--243,
  2003.

\bibitem{O}
R.~Olkiewicz.
\newblock Environment-induced superselection rules in markovian regime.
\newblock {\em Comm. Math. Phys.}, 208(1):245--265, Dec 1999.

\bibitem{O2}
R.~Olkiewicz.
\newblock Structure of the algebra of effective observables in quantum
  mechanics.
\newblock {\em Ann. Physics}, \textbf{286}(1):10 -- 22, 2000.

\bibitem{TV}
F.~Ticozzi and L.~Viola.
\newblock Quantum markovian subsystems: Invariance, attractivity, and control.
\newblock {\em IEEE Trans. Automat. Control}, 53(9):2048--2063, Oct 2008.

\bibitem{Hellmich}
M.~Hellmich.
\newblock Quantum dynamical semigroups and decoherence.
\newblock {\em Adv. Math. Phys.}, 2011:16, 2011.

\bibitem{CSU2}
R.~Carbone, E.~Sasso, and V.~Umanit{\`a}.
\newblock Decoherence for quantum markov semi-groups on matrix algebras.
\newblock {\em Ann. Henri Poincar{\'e}}, \textbf{14}(4):681--697, 2013.

\bibitem{CSU-new}
R.~Carbone, E.~Sasso, and V.~Umanit{\`a}.
\newblock Environment induced decoherence for markovian evolutions.
\newblock {\em J. Math. Phys.}, \textbf{56}(9):092704, 2015.

\bibitem{DFR}
A.~Dhahri, F.~Fagnola, and R.~Rebolledo.
\newblock The decoherence-free subalgebra of a quantum markov semigroup with
  unbounded generator.
\newblock {\em Infin. Dimens. Anal. Quantum Probab. Relat. Top.},
  13(03):413--433, 2010.

\bibitem{FFRR}
F.~Fagnola and R.~Rebolledo.
\newblock Algebraic conditions for convergence of a quantum markov semigroup to
  a steady state.
\newblock {\em Infin. Dimens. Anal. Quantum Probab. Relat. Top.},
  11(03):467--474, 2008.

\bibitem{FV82}
A.~Frigerio and M.~Verri.
\newblock Long-time asymptotic properties of dynamical semigroups on
  $w^*$-algebras.
\newblock {\em Math. Z.}, \textbf{180}(3):275--286, Sep 1982.

\bibitem{DFSU}
J.~Deschamps, F.~Fagnola, E.~Sasso, and V.~Umanit{\`a}.
\newblock Structure of uniformly continuous quantum markov semigroups.
\newblock {\em Rev. Math. Phys.}, \textbf{28}(01):1650003, 2016.

\bibitem{FSU-NtFt}
F.~Fagnola, E.~Sasso, and V.~Umanit\`a.
\newblock The role of the atomic decoherence-free sub- algebra in the study of
  quantum markov semigroups.
\newblock {\em J. Math. Phys.}, 60, 2019.

\bibitem{batkai}
A.~Batkai, U.~Groh, U.~Kunszenti-Kovacs, and M.~Schreiber.
\newblock Decomposition of operator semigroups on $w^*$-algebras.
\newblock {\em Semigroup Forum}, 84:8--24, 2012.

\bibitem{Partha}
K.~R. Parthasarathy.
\newblock {\em An introduction to quantum stochastic calculus}, volume~85 of
  {\em Monographs in Mathematics}.
\newblock Birkh{\"a}user, 1992.

\bibitem{SG}
K.~B. Sinha and D.~Goswami.
\newblock {\em Quantum Stochastic Processes and Noncommutative Geometry}.
\newblock Cambridge Tracts in Mathematics. Cambridge University Press, 2007.

\bibitem{Evans}
D.~E. Evans.
\newblock Irreducible quantum dynamical semigroups.
\newblock {\em Comm. Math. Phys.}, \textbf{54}(3):293--297, Oct 1977.

\bibitem{kadison}
R.~V. Kadison and J.~R. Ringrose.
\newblock {\em Fundamentals of the theory of operator algebras}, volume~2 of
  {\em Graduate Studies in Mathematics}.
\newblock American Mathematical Society, 1997.

\bibitem{tomi3}
J.\ Tomiyama.
\newblock On the projection of norm one in w$^*$-algebras iii.
\newblock {\em T\^ohoku Math. J.}, 11:125--129, 1959.

\bibitem{Black}
B.~\ Blackadar.
\newblock {\em Operator Algebra}, volume 121.
\newblock Encyclopaedia of Mathematical Sciences, 2005.

\bibitem{OL}
P.~Lugiewicz and R.~Olkiewicz.
\newblock Classical properties of infinite quantum open systems.
\newblock {\em Comm. Math. Phys.}, 208(239):241--259, Jan 1999.

\bibitem{CJ}
R.~Carbone and A.~Jen{\v c}ov\'a.
\newblock On period cycles and fixed points of a quantum channel.
\newblock {\em Ann. Henri Poincar\'e}, 1:155--188, 2020.

\end{thebibliography}
\end{document}